\documentclass[12pt]{article}
\usepackage[utf8]{inputenc}
\usepackage[letterpaper,margin=1in]{geometry}
\usepackage[fleqn]{amsmath}
\usepackage{amssymb}
\usepackage{amsthm}
\usepackage{thmtools}
\usepackage{extarrows}
\usepackage{enumerate}
\usepackage{graphicx}
\usepackage{indentfirst}
\usepackage{natbib}
\usepackage{verbatim}
\usepackage{booktabs}
\usepackage{xcolor}
\usepackage{titlesec}
\usepackage{longtable}
\usepackage{threeparttable}
\usepackage{bookmark}

\newtheorem{theorem}{Theorem}
\newtheorem{lemma}{Lemma}

\newtheorem{assumption}{Assumption}

\let\oldproofname=\proofname
\renewcommand{\proofname}{\rm\bf{\oldproofname}}
\allowdisplaybreaks

\title{Two-step Estimation of Network Formation Models with
Unobserved Heterogeneities and Strategic Interactions}
\author{Shaomin Wu}

\begin{document}

\abovedisplayskip=15pt
\belowdisplayskip=15pt

\maketitle
\centerline{\textbf{Abstract}}
\indent In this paper, I characterize the network formation process as a static game of incomplete information, where the latent payoff of forming a link between two individuals depends on the structure of the network, as well as private information on agents' attributes. I allow agents' private unobserved attributes to be correlated with observed attributes through individual fixed effects. Using data from a single large network, I propose a two-step estimator for the model primitives. In the first step, I estimate agents' equilibrium beliefs of other people's choice probabilities. In the second step, I plug in the first-step estimator to the conditional choice probability expression and estimate the model parameters and the unobserved individual fixed effects together using Joint MLE. Assuming that the observed attributes are discrete, I showed that the first step estimator is uniformly consistent with rate $N^{-1/4}$, where $N$ is the total number of linking proposals. I also show that the second-step estimator converges asymptotically to a normal distribution at the same rate.

\section{Introduction}
The social network is an important feature to take into account when studying many economic behaviors, from peer effects in education and crime to the dynamics of product adoption and financial contagions. However, most network studies of these behaviors are challenged by the endogeneity of the network. This highlights the importance of developing econometric models of network formation. Moreover, the network formation process is itself an interesting subject to study, since it sheds light on real-world behaviors such as how people engage with each other on social media platforms. \\
\indent Two features are crucial in a network formation model. The first feature is strategic interactions. The incentives of forming a link in a network are not only affected by the two agents' characteristics but also the linking decisions of other agents, such as the "popularity effect"-- an agent is more likely to form a link with another agent who has many friends. The second feature involves unobserved agent-level heterogeneities, which are typically private information that is known only to the agent themselves, such as an individual's personality traits on a dating app. The agent-level unobserved heterogeneities are correlated with observed characteristics but are unobserved to other agents or researchers. Depicting these two features is essential for effectively modeling the network formation process and accurately inferring agents' preferences. Motivated by this, I study a directed network formation model with individual-specific unobserved heterogeneities and strategic interactions. In the incremental utility of a link from person $i$ to $j$, I include the linking choices of the person $j$ to capture the popularity effect, and include individual fixed effects to capture agent-level unobserved heterogeneities, while remaining agnostic about the conditional distribution of the agent-level unobservables, not requiring it to be known to the researchers.\\ 
\indent There's growing literature on the estimation of network formation models. Among them, this paper is most related to \cite{leung2015two} and \cite{ridder2022estimation}. Both of them study the estimation of network formation games with incomplete information and strategic interactions and assume that the private information is independent of observed characteristics. \cite{leung2015two} lets the payoff depend on network structure in a separable way, through the sum of incremental utilities from each link. Then the optimal link choices are myopic, in the sense that an agent chooses to form a link with another member if the expected utility of forming that link is greater than 0. To be specific, let $G_{ij}$ denote the linking proposal from individual $i$ to $j$, and let $X_i, X_j$ denote observed characteristics of the two individuals; let $\epsilon_{ij}$ denote unobserved link-specific characteristics that are independent with $X$. \cite{leung2015two}'s model yields the following optimal linking decision:
\begin{align}
    G_{ij}=\mathbf{1}\left\{w(X_i,X_j)\beta_0+\mathbb{E}[G_{-ij}|X,\sigma]+\varepsilon_{ij}\geqslant 0\right\},\label{leung linking decision}
\end{align}
where $w$ is a known function capturing the homophily effect, and $\sigma$ denotes the equilibrium.
\cite{ridder2022estimation} considers a more general case in which the utility function depends on the choice of potential partners in a non-separable way, for example, allowing the utility to depend on links-in-common. Using the Legendre transform, they show that even under this general case, the optimal linking choice is still equivalent to a sequence of myopic link choices. For estimation, both of the two papers assume that the data observed comes from a symmetric equilibrium, whereby agents with the same observable characteristics have the same equilibrium linking probabilities, i.e. $P(G_{ij}|X_{ij}=x, X)=P(G_{kl}|X_{kl}=x, X)$. Then the conditional linking probabilities can be estimated in the first step, by taking the empirical frequency with which agents with the same observable characteristics link to each other. In terms of strategic interactions, this paper adopts the same framework as \cite{leung2015two}, including only the popularity effect and keeping the dependence on the network structure to be separable, which is simpler than \cite{ridder2022estimation}'s framework. Different from the two papers, this paper studies the case when private information is correlated with observables by including individual fixed effects in the utility. For estimation, this paper also adopts a two-step procedure and estimates the realized equilibrium beliefs in the first step. This allows us to circumvent the difficulty to specify the equilibrium selection mechanism when there might be multiple equilibria.\\
\indent This paper is also closely related to \cite{graham2017econometric}, which studies a network formation model with dyadic link formation. In their model, the linking decision between individual $i$ and $j$ only depends on the characteristics of $i$ and $j$ and there are no strategic interactions. Let $A_i, A_j$ denote individual fixed effects unobserved to researchers. The linking decision in \cite{graham2017econometric} is
\begin{align}
    G_{ij}=\mathbf{1}\left\{w(X_i,X_j)\beta_0+A_i+A_j+\varepsilon_{ij}\geqslant0\right\}\label{graham linking decision}
\end{align}
Same as \cite{graham2017econometric}, this paper also incorporates unobserved individual fixed effects. The difference is that my model contains strategic interactions, so the information structure matters. I assume that individual fixed effects $A_i$ are private information that is i.i.d. across individuals. The agents know the distribution of the individual fixed effects so that they can form beliefs of the expected "type"\footnote{I don't assume $A_i$ to have discrete distribution, though.} of other people. From the modeling point of view, this paper studies a model which is a combination of (\ref{leung linking decision}) and (\ref{graham linking decision}). Note that a special case of (\ref{leung linking decision}) is when $\epsilon_{ij}$ can be written as the sum of an individual "random effect" $A_i$ and an idiosyncratic error $\nu_{ij}$. This is different from this paper's setting since $A_i$ is assumed to be independent with $X$ in \cite{leung2015two}.\\
\indent Another strand of literature on estimating strategic network formation models assumes complete information, such as \cite{miyauchi2016structural} and  \cite{sheng2020structural}. These models are the hardest to deal with because they generally admit multiple equilibria and thus achieve set but not point identification of the model parameters. This paper shies away from these cases by assuming incomplete information.\\
\indent The rest of the paper is organized as follows. In section 2, I develop the model and derive the optimal link choices. In section 3, I propose a two-step estimation procedure and show  the consistency of the first-step estimator. In section 4, I showed the asymptotic distribution of the estimators. The last section concludes.
\section{The Model}
I consider the directed network formation model in this paper. The formation process is a static game of incomplete information. An agent's payoff of forming a link depends on idiosyncratic private information. Given the belief of other people's linking decisions, agents form their own links simultaneously. Formally, the network formation game is set up as follows:\\
\indent There are $n$ agents indexed by $i\in \mathcal I=\{1,2,...,n\}$. Each agent chooses whether or not to link with the other $n-1$ agents. Player $i$'s action vector $G_i=(G_{i1}, G_{i2},...,G_{ij},..., G_{in})'$ where $j\neq i$ is chosen from the action profile ${A}$ which has $2^{n-1}$ components. 
The payoff function of individual $i$ is 
\begin{align}
U_i(G, X, A_i, \varepsilon_i)=\sum _{j=1}^nG_{ij}\left(u_{ij}(G_{-i},X,A_i;\beta)+\varepsilon_{ij}\right).\label{payoff}
\end{align}
The deterministic part of incremental utility from link $ij$ is specified as 
\begin{align}
    u_{ij}(G_{-i},X,A_i;\beta)=w(X_i,X_j)\beta_1+A_i+G_{ji}\beta_2+\frac{1}{n-2}\sum_{k\neq i,j} G_{jk}\beta_3.
\end{align}
where the first term captures the homophily effect. $w$ is a known function. $X=(X_1',..., X_n')'$ is public information for all agents and is observable to researchers. For simplicity, write $w(X_i,X_j)=W_{ij}$ from now on. The second term $A_i$ is individual-specific heterogeneity, which is unobserved both to other agents and researchers. Let $F_{A|X}$ be the distribution of $A_i$ conditional on observables, which is assumed to be independent and identical across $i$, and known to all agents, but not necessarily known to researchers. $A_i$ can be correlated with $X$. The third and the last term capture the popularity effect. The realization of $\varepsilon_i=(\varepsilon_{i1},...,\varepsilon_{in})'$ is agent $i$'s private information which is also unobserved to researchers. The model is therefore a static game with incomplete information, and the solution concept is Bayesian Nash Equilibrium. Different from \cite{leung2015two}, my model allows private information to be correlated with common information while doesn't require the conditional distribution of private information to be known to researchers. Also, I allow "asymmetric" equilibrium which will be mentioned later in this part.\\
\indent For the above model, I impose the following assumptions:
\begin{assumption} \label{independence}
    (a) $X_i\perp X_j$ for $i\neq j$. $X_i$ is discrete distributed with finite support $\mathbb{X}=\{x_1,...,x_{T_x}\}$. (b) $A_i$ are independently and identically distributed. The conditional CDF $F_{A|X}$ is known to all agents but unknown to researchers. (c) $\varepsilon_{ij}$ are i.i.d. with logit distribution $F_{\varepsilon}$, which is known to both agents and researchers. (d) $\varepsilon_i\perp\ (X', A)'$ for all $i$. 
\end{assumption}
\indent Let $\delta_j(X,A_j,\varepsilon_j)$ denote agent $j$'s (pure) strategy. Let $\sigma_j(a|X, A_j)= Pr\left(\delta_j(X,A_j,\varepsilon_j)=a|X,A_j\right)$ denote the agent $i$'s belief that agent $j$ of type $A_j$ chooses action $a$, given commonly known information $X$ and agent $i$'s private information. By Assumption \ref{independence} (b) and (c), actions $G_i$ and $G_j$, $i\neq j$ are independent given commonly known attributes $X$. This fact simplifies the proof of consistency by weakening the correlation between links. Since agent $i$ actually doesn't known the realization of $A_j$, so agent $i$'s expected utility from choosing action $g_i\in S$ is $\sum_{g_{-i}}
U_i(g_i,g_{-i}, X, A_i, \varepsilon_i)\mathbb E_{A_{-j}}\big[\sigma_{-i}(g_{-i}|X, A_{-j})\big]$. Therefore,
\begin{align*}
Pr&(G_i=g_i|X,A_i,\sigma)\\
&=Pr\Bigg(\sum_{g_{-i}}
\bigg[U_i(g_i,g_{-i}, X, A_i, \varepsilon_i)-U_i(\tilde g_i,g_{-i}, X, A_i, \varepsilon_i)\bigg]\mathbb E_{A_{-j}}\big[\sigma_{-i}(g_{-i}|X, A_{-j})\big]>0,\\
&\qquad\qquad \forall \tilde g_i\in S\Bigg| X, A_i,\sigma\Bigg)
\end{align*}
A (Bayesian) equilibrium $\sigma^*(X, A_i)$ is a belief function that solves the fixed point equation:
\begin{align*}
    \sigma_i^*(a|X, A_i)=Pr(G_i=a|X,A_i,\sigma^*)
\end{align*}
for all $X\in \mathbf{X}$, agents $i\in \mathcal I$ and actions $a\in S$.\\
\indent I consider "symmetric" equilibria in which pairs of agents with the same observable attributes and the same type ($A_i$) have the same conditional linking probabilities. For any $(X,A_i,\varepsilon_{ij})$ and "symmetric" belief profile $\sigma$ in a neighborhood of an "symmetric" equilibrium $\sigma^*$ , player $i$'s optimal strategy $G_{i}(X, A_i, \varepsilon_i, \sigma)=\big(G_{ij}(X, A_i, \varepsilon_i, \sigma)\big)_{j\neq i}$ is given by:
\begin{align}
    G_{ij}(X, A_i, \varepsilon_i, \sigma)&=\mathbf{1}\Big\{\mathbb E\big[u_{ij}(G_{-i}, X, A_i;\beta)\big|X,A_i, \sigma\big]+\varepsilon_{ij}\geqslant0\Big\}\label{optimal strategy}.
\end{align}
\indent Assuming a symmetric equilibrium exists, the model is incomplete because there could be multiple equilibria for any realization of $(X, A, \varepsilon)$. For completeness of the model, I specify the equilibrium selection mechanism in the following assumption. The mechanism, however, is not explicitly used in writing the likelihood function in part 3, because by using two-step estimation, I can avoid specifying the equilibrium theoretically. For the convenience of defining equilibrium selection mechanisms, I add subscript $n$ to $G$, $X$, $A$, and $\varepsilon$. The equilibrium selection mechanism is a measurable function $\lambda_n: (X_n, \nu_n,\beta_0) \mapsto \sigma_n\in \mathcal{G}(X_n, A_n, \beta_0)$, where $\mathcal{G}(X_n, A_n, \beta_0)$ is the set of symmetric equilibria.
\begin{assumption}\label{Equilibrium Selection}
    \textbf{(Equilibrium Selection)} There exist sequences of equilibrium selection mechanisms $\{\lambda_n(\cdot); n\in\mathbb N\}$ and public signals $\{\nu_n;n\in \mathbb N\}$ such that for $n$ sufficiently large, $\mathcal{G}(X_n, \beta_0)$ is nonempty, and for any $g_n\in S^n$, 
    \begin{align*}
        Pr(G_n=\mathbf g_n|X_n, A_n)=\sum_{\sigma_n\in \mathcal{G}(X_n, A_n, \beta_0)=\sigma_n|X_n, A_n} Pr(\lambda(X_n, \nu_n;\beta_0)=\sigma_n|X_n, A_n)\prod_{i=1}^n\sigma_{i}(g_{i}|X_n, A_{i})
    \end{align*}    
\end{assumption}
\section{Estimation}
\indent  Define $P_{ij}(X,A_{i},\sigma)$ to be the probability that individual $i$ proposes to form a link with $j$ conditional on $X$ $A_{i}$, and $\sigma$. According to (\ref{optimal strategy}) and Assumption \ref{independence} (c),
\begin{align}
    &P_{ij}(X, A_{i,n} ,\sigma)=Pr(G_{ ij}(X,A_{i}, \varepsilon_{i},\sigma)=1|X, A_{i}, \sigma)\notag\\
    =&\frac{\exp\Big(W_{ij}\beta_0+A_{i}+\mathbb E_{A_j}\big[\sigma_{ji}(G_{ji}=1\big|X,A_j)\big]\beta_1+\frac{1}{n-2}\sum_{k\neq i,j}\mathbb E_{A_j}\big[\sigma_{jk}(G_{jk}=1\big|X,A_j)\big]\beta_2\Big)}{1+\exp\Big(W_{ij}\beta_0+A_{i}+\mathbb E_{A_j}\big[\sigma_{ji}(G_{ji}=1\big|X,A_j)\big]\beta_1+\frac{1}{n-2}\sum_{k\neq i,j}\mathbb E_{A_j}\big[\sigma_{jk}(G_{jk}=1\big|X,A_j)\big]\beta_2\Big)}\label{ccp}
\end{align}
Define $p_{ij}(X, A_{i})$ to be the equilibrium probability that agent $i$ proposes a link to agent $j$. which is realized in the data. Equilibrium condition requires that
\begin{align}
    &p_{ij}(X, A_{i})=P_{ij}\left(X,  A_{i}, p_{}(X, A_{i})\right)\notag\\
    =&\frac{\exp\Big(W_{ij}\beta_0+A_{i}+\mathbb E_{A_{j}|X}\big[p_{ji}(X,A_{j})\big]\beta_1+\frac{1}{n-2}\sum_{k\neq i,j}\mathbb E_{A_{j}|X}\big[p_{jk}(X,A_{j})\big]\beta_2\Big)}{1+\exp\Big(W_{ij}\beta_0+A_{i}+\mathbb E_{A_{j}|X}\big[p_{ji}(X,A_{j})\big]\beta_1+\frac{1}{n-2}\sum_{k\neq i,j}\mathbb E_{A_{j}|X}\big[p_{jk}(X,A_{j})\big]\beta_2\Big)}\label{equilibrium condition}
\end{align}
\indent For notation simplicity, denote $q_{jk}(X, \sigma^*):=\mathbb E_{A_{j}}\big[Pr(G_{ jk}(X,A_{j}, \varepsilon_{j},\sigma)=1|X, A_{j}, \sigma^*)\big]$, which is the probability that agent $j$ proposes a link to $k$ conditional on $X$ and the realized equilibrium $\sigma^*$. Then (\ref{equilibrium condition}) can be rewritten as 
\begin{align}
    &P_{ij}\left(X,  A_{i}, p(X, A_{i})\right)\notag\\
    =&\frac{\exp\Big(W_{ij}\beta_0+A_{i}+q_{ji}(X, \sigma^*)\beta_1+\frac{1}{n-2}\sum_{k\neq i,j}q_{jk}(X, \sigma^*)\beta_2\Big)}{1+\exp\Big(W_{ij}\beta_0+A_{i}+q_{ji}(X, \sigma^*)\beta_1+\frac{1}{n-2}\sum_{k\neq i,j}q_{jk}(X, \sigma^*)\beta_2\Big)}\notag\\
    :=& Q_{ij}\left(X,  A_{i}, q_{}(X, \sigma^*)\right)\label{new notation ccp}
\end{align}
\indent Although $p_{ jk}(X, A_{j})$ is not identified from data, $q_{ij}(X)$ is identified. With abuse of notations, let $q_{ st}(X)=\mathbb E_{A_{j}}\big[p_{ jk}(X_{j}=x_s, X_{k}=x_t, X,A_{j})\big]$.
\\
\indent Consider the empirical frequency of pairs with the same observable characteristics proposing to form a link:
\begin{align*}
    \hat q_{n,st}=\frac{\sum_i\sum_{j\neq i}G_{ij}\mathbf{1}\big\{X_{i}=x_s,X_{j}=x_t\big\}}{\sum_i\sum_{j\neq i}\mathbf{1}\big\{X_{i}=x_s,X_{j}=x_t\big\}}
\end{align*}

First, I want to show that $q_{st}(X, \sigma^*)$ can be consistently estimated by $\hat q_{n, st}$ under the payoff function specified in \ref{payoff}. Formally, I want to prove the following lemma:
\begin{lemma}\label{qhat}
    For any $X$ and realized symmetric equilibrium $\sigma^*$,
    \begin{align*}
        \sup_{s,t}\big|\hat q_{n, st}-q_{st}(X, \sigma^*)\big|=O_p\big(n^{-1/2}\big).
    \end{align*}
\end{lemma}
\begin{proof}
    See the Appendix.
\end{proof}
\indent For the convenience of the following analysis, I introduce a change of notation:
\begin{align*}
    Z_{ij}:=(W_{ij}', q_{ji}, \frac{1}{n-2}\sum_{k\neq i,j} q_{jk})'
\end{align*}
and 
\begin{align*}
    \hat Z_{ij}:=(W_{ij}', \hat q_{ji},  \frac{1}{n-2}\sum_{k\neq i,j} \hat q_{jk})'
\end{align*}
then by Lemma \ref{qhat}, $\sup_{s,t}|\hat Z_{s,t}- Z_{s,t}|=O_p\big(n^{-1/2}\big)$\\
\indent With the estimates $\hat q_n=\{\hat q_{st}\}_{\forall s,t}$, I propose to estimate the parameter $\beta$ and individual fixed effects $\{A_i\}_{i=1}^n$ jointly by MLE.\\
\indent By Assumption \ref{independence} (c), the conditional likelihood of the network is 
\begin{align*}
    P(G=\mathbf g|X, A)=\prod_{i\neq j} Pr(G_{ ij}(X,A_{i}, \varepsilon_{i},\sigma)=g|X, A_{i}, \sigma)
\end{align*}
\indent By (\ref{ccp}) and (\ref{new notation ccp}),
\begin{align*}
    &Pr(G_{ ij}(X,A_{i}, \varepsilon_{i},\sigma)=g|X, A_{i}, \sigma)\\
    =&Q_{ij}(X,A,q(X))^{g}\big[1-Q_{ij}(X,A,q(X))\big]^{1-g}
\end{align*}

Construct the log-likelihood function:
\begin{align}
    \mathcal{L}_n(\beta,A,q)=\frac{1}{n(n-1)}\sum_i\sum_{j\neq i}G_{ij}\ln Q_{ij}(\beta, A_{i}, q)+(1-G_{ij})\ln (1-Q_{ij}(\beta, A_{i}, q))\label{loglikelihood}
\end{align}
Let $\hat \beta$ and $\hat A$ be the maximizer of the log-likelihood with $q$ replaced by $\hat q_n$.
\begin{align*}
    \max_{\beta, A} \mathcal{L}_n(\beta, A, \hat q_n).
\end{align*}
By first concentrating out $A$, the estimators are given by:
\begin{align}
    \hat \beta = \arg\max_{\beta} \mathcal{L}^c_n(\beta, \hat A(\beta), \hat q_n)\label{conlikelihood}
\end{align}
where 
\begin{align*}
    &\hat A(\beta)=\arg\max_{A} \mathcal{L}_n(\beta, A, \hat q_n)\\
    \Longrightarrow&\hat A_i(\beta)=\arg\max_{A_i}\frac{1}{n-1}\sum_{j\neq i}G_{ij}\ln Q_{ij}(\beta, A_{i}, \hat q_n)+(1-G_{ij})\ln (1-Q_{ij}(\beta, A_{i}, \hat q_n))
\end{align*}
\indent By rearranging the sample score of (\ref{loglikelihood}), it can be shown that $\hat A(\beta)$, when it exists, is the unique solution to the fixed point problem:
\begin{align}
    \hat A(\beta)=\varphi(\hat A(\beta))
\end{align}
where 
\begin{align}
    \varphi(A)=\begin{pmatrix}\ln \sum_{j\neq 1}G_{1j}-\ln\sum_{j\neq 1}\frac{\exp(\hat Z_{1j}'\beta)}{1+\exp(\hat Z_{1j}'\beta+A_1)}\\\vdots\\\ln \sum_{j\neq n}G_{nj}-\ln\sum_{j\neq n}\frac{\exp(\hat Z_{nj}'\beta)}{1+\exp(\hat Z_{nj}'\beta+A_n)}\end{pmatrix}\label{fixedpointsolution}
\end{align}
\section{Asymptotic Analysis}
\indent In this part, I first show the consistency of $\hat\beta$ and $\hat A$. Because link proposals from the same individual are correlated, the first step estimator has a slow convergence rate $\sqrt{n}$, which is equivalent to the usual convergence rate of $N^{1/4}$, since the number of summands in the likelihood function $N=n(n-1)$.  As is well discussed in the nonlinear panel literature, there could be an estimation bias of $\hat\beta$ caused by the incidental parameters problem (e.g. \cite{hahn2004jackknife}, \cite{arellano2007understanding}). However, as I will show in this part, the effect of second-step bias is dominated by the slow convergence rate of the first step, so a bias term won't show up in the asymptotic distribution.
\begin{assumption}\label{Compact Support}
    \textbf{(Compact Support)} $\beta_0\in \text{int} (\mathbb B)$, with $\mathbb B$ a compact subset of $\mathbb R^K$.
\end{assumption}
\begin{assumption}\label{Joint FE Identification}
    \textbf{(Joint FE Identification)} $\mathbb E[\mathcal L_n(\beta,A,q)|X,A_0]$ is uniquely maximized at $\beta=\beta_0$ and , $A=A_0$, for large enough n.
\end{assumption}
\indent Compactness of the support (Assumption \ref{independence} (a)(b) and Assumption \ref{Compact Support}) implies that 
\begin{align}
    Q_{ij}(\beta, A_i, q) \in (\kappa, 1-\kappa)\label{bdd}
\end{align}
for some $0<\kappa<1$ and for all $A_i\in\mathbb A$, $\beta\in \mathbb B$ and $\forall q\in (k, 1-k)$.
\begin{theorem}\label{Consistency}
    \textbf{(Consistency)} Under Assumptions \ref{independence}, \ref{Equilibrium Selection}, \ref{Compact Support}, and \ref{Joint FE Identification}
    \begin{align*}
        &\hat\beta\xrightarrow{p}\beta_0;\\
        &\hat A\xrightarrow{p} A_0.
    \end{align*}
\end{theorem}
\begin{proof}
    See the Appendix.
\end{proof}
\indent With a more involved argument, I can actually show the uniform convergence rate of $\hat A$
\begin{theorem} \label{uniform A}
     With probability $1-O(n^{-2})$,
    \begin{align*}
         \sup_{1\leqslant i\leqslant n}|\hat A_i-A_{i0}|< O\left(\sqrt{\frac{\ln n}{n}}\right).
    \end{align*}
\end{theorem}
\begin{proof}
    See the Appendix.
\end{proof}
To state the form of the asymptotic distribution, define
\begin{align}
    \mathcal I_0&=\lim_{n\rightarrow \infty}-\frac{1}{n(n-1)}\sum_{i=1}^n\sum_{j\neq i}Z_{ij}Z_{ij}'Q_{ij}(1-Q_{ij})\notag\\
   &+\frac{1}{n}\sum_{i=1}^n\frac{\left(\frac{1}{n-1}\sum\limits_{j\neq i}Q_{ij}(1-Q_{ij})Z_{ij}\right)\left(\frac{1}{n-1}\sum\limits_{j\neq i}Q_{ij}(1-Q_{ij})Z_{ij}'\right)}{\frac{1}{n-1}\sum_{j\neq i}Q_{ij}(1-Q_{ij})}\label{i0},
\end{align}
\begin{theorem}\label{asymptotic normality}
    Under Assumptions \ref{independence}, \ref{Equilibrium Selection}, \ref{Compact Support}, and \ref{Joint FE Identification}, 
    \begin{align*}
        \frac{\sqrt{n}a'(\hat \beta-\beta_0)}{\|a\|^{-1/2}(a'\mathcal{I}_0^{-1}\Omega_n\mathcal{I}_0^{-1}a)^{1/2}}\xrightarrow{d}N(0,1)
    \end{align*}
for any $d\times 1$ vector of real constants $a$ and $\Omega_n$ as defined in the Appendix.
\end{theorem}
\begin{proof}
    See the Appendix.
\end{proof}
\section{Monte Carlo Simulation}
In this section, I implement the proposed method in a simulation study. Assume the following utility specification:
\begin{align*}
     U_i(G, X, A_i, \varepsilon_i)=\sum _{j=1}^nG_{ij}\left(|X_i-X_j|\beta_1+A_i+G_{ji}\beta_2+\frac{1}{n-2}\sum_{k\neq i,j} G_{jk}\beta_3+\varepsilon_{ij}\right)
\end{align*}
where $X_i$ is a random variable taking values in $\{1,-1\}$ with equal probability, and $\epsilon_{ij}$ follows the Logistic distribution. The distribution of $A_i$ is generated according to
\begin{align*}
    A_i = (\alpha_L+\gamma a_i)\mathbf{1}\{X_i=-1\}+\alpha_H\mathbf{1}\{X_i=1\}+V_i,
\end{align*}
with $\alpha_L< \alpha_H$ and $a_i\sim N(0, 0.1), V_i\sim N(0, \sqrt{0.1})$, and they are independent. In the simulation exercise, I consider three scenarios. In the first two scenarios, $A_i$ is correlated with $X$. In Scenario 1, I let $\alpha_L=-2/3, \alpha_H=-1/6$, and $\gamma=0$, so that the correlation between $A_i$ and $X_i$ is only through the value of $X_i$. In Scenario 2, I let $\alpha_L=-2/3, \alpha_H=-1/6$, and $\gamma=1$, so that the correlation between $A_i$ and $X_i$ is determined not only by the value of $X_i$ but also by the identity of $i$ (captured by the random variable $a_i$). In Scenario 3, I let $\alpha_L=-1/2, \alpha_H=-1/2$ and $\gamma=0$, so that $A_i$ is independent with $X$. The true values of the parameters are $(\beta_1, \beta_2, \beta_3)=(-2,1,1)$. The network is generated according to the $n-$player incomplete information game described in Section 2, with $n$ taking values of $50, 100, 250$, and $500$. For each value of $n$, I generate a single network and use the method proposed in this paper and \cite{leung2015two} to estimate the parameters. When using \cite{leung2015two}'s estimator, the private information $\eta_{ij}$ is the sum of $A_i$ and $\epsilon_{ij}$ with $A_i\perp\epsilon_{ij}$. Each experiment is repeated $1000$ times. I report the means and standard errors of the estimated parameters in the tables below.
\begin{table}[htbp]
  \centering
  \caption{Scenario 1 Correlated Private Information ($\alpha_L=-2/3, \alpha_H=-1/6,\gamma = 0$)}
    \begin{tabular}{rccccccc}
    \toprule
          & \multicolumn{3}{c}{This paper's estimator} &       & \multicolumn{3}{c}{\cite{leung2015two}'s estimator} \\
    \multicolumn{1}{c}{$n$} & $\beta_1$ & $\beta_2$ & $\beta_3$ &       & $\beta_1$ & $\beta_2$ & $\beta_3$ \\
    \midrule
    \multicolumn{1}{c}{50} & -1.922 & 0.966 & 0.936 &       & -2.092 & 0.827 & 1.390 \\
          & (0.049) & (0.101) & (0.069) &       & (0.249) & (0.484) & (1.024) \\
    \multicolumn{1}{c}{100} & -1.930 & 0.951 & 0.926 &       & -2.085 & 0.807 & 1.435 \\
          & (0.035) & (0.058) & (0.033) &       & (0.198) & (0.478) & (0.985) \\
    \multicolumn{1}{c}{250} & -1.967 & 1.050 & 0.973 &       & -2.065 & 0.864 & 1.334 \\
          & (0.042) & (0.101) & (0.055) &       & (0.170) & (0.476) & (0.961) \\
    \multicolumn{1}{c}{500} & -2.015 & 1.065 & 0.975 &       & -2.047 & 0.919 & 1.237 \\
          & (0.036) & (0.057) & (0.035) &       & (0.160) & (0.476) & (0.950) \\
    \bottomrule
    \end{tabular}
    \begin{minipage}{12cm}
    \vspace{0.1cm}
    \scriptsize
    \linespread{1}
    This table gives the mean of each estimator across the 1000 Monte Carlo estimates. The standard deviation of the Monte Carlo estimates is reported below the mean value of the point estimates in parentheses (this is a quantile-based estimate which uses the 0.05 and 0.95 quantiles of the Monte Carlo distribution of point estimates and the assumption of Normality).
    \end{minipage}
  \label{desig 1}%
\end{table} 

\begin{table}[htbp]
  \centering
  \caption{Scenario 2 Correlated Private Information ($\alpha_L=-2/3, \alpha_H=-1/6,\gamma=1$)}
    \begin{tabular}{rccccccc}
    \toprule
          & \multicolumn{3}{c}{This paper's estimator} &       & \multicolumn{3}{c}{\cite{leung2015two}'s estimator} \\
    \multicolumn{1}{c}{$n$} & $\beta_1$ & $\beta_2$ & $\beta_3$ &       & $\beta_1$ & $\beta_2$ & $\beta_3$ \\
    \midrule
    \multicolumn{1}{c}{50} & -1.921 & 0.952 & 0.928 &       & -2.072 & 0.866 & 1.294 \\
          & (0.043) & (0.108) & (0.069) &       & (0.280) & (0.698) & (1.521) \\
    \multicolumn{1}{c}{100} & -1.932 & 0.909 & 0.902 &       & -2.079 & 0.820 & 1.408 \\
          & (0.048) & (0.034) & (0.021) &       & (0.258) & (0.732) & (1.539) \\
    \multicolumn{1}{c}{250} & -1.955 & 0.956 & 0.925 &       & -2.056 & 0.888 & 1.287 \\
          & (0.045) & (0.055) & (0.029) &       & (0.242) & (0.738) & (1.509) \\
    \multicolumn{1}{c}{500} & -2.015 & 1.023 & 0.953 &       & -2.035 & 0.956 & 1.162 \\
          & (0.047) & (0.069) & (0.038) &       & (0.233) & (0.724) & (1.460) \\
    \bottomrule
    \end{tabular} \\
    \begin{minipage}{12cm}
    \vspace{0.1cm}
    \scriptsize
    \linespread{1}
    This table gives the mean of each estimator across the 1000 Monte Carlo estimates. The standard deviation of the Monte Carlo estimates is reported below the mean value of the point estimates in parentheses (this is a quantile-based estimate which uses the 0.05 and 0.95 quantiles of the Monte Carlo distribution of point estimates and the assumption of Normality).
    \end{minipage}
  \label{desig 2}%
\end{table}

\indent As can be seen in Table \ref{desig 1} and \ref{desig 2}, when the private information is correlated with observed individual characteristics $X$, this paper's approach yields good estimates for the parameters, while \cite{leung2015two}'s estimator doesn't perform well, both in terms of the mean and variance of the estimators. This is not surprising since \cite{leung2015two} assumes that private information and observable individual characteristics are independent. Under the correlated scenario, \cite{leung2015two}'s estimator will not be consistent. Table \ref{desig 3} shows the simulation results when the individual private information $A$ is independent with observed characteristics $X$. Not surprisingly, both this paper's estimator and \cite{leung2015two}'s estimator perform reasonably well, except that \cite{leung2015two}'s estimator has  larger variances.
\begin{table}[htbp]
  \centering
  \caption{Scenario 3 Independent Private Information ($\alpha_L=-1/2, \alpha_H=-1/2, \gamma=0$)}
    \begin{tabular}{rccccccc}
    \toprule
          & \multicolumn{3}{c}{This paper's estimator} &       & \multicolumn{3}{c}{\cite{leung2015two}'s estimator} \\
    \multicolumn{1}{c}{$n$} & $\beta_1$ & $\beta_2$ & $\beta_3$ &       & $\beta_1$ & $\beta_2$ & $\beta_3$ \\
    \midrule
    \multicolumn{1}{c}{50} & -1.913 & 0.951 & 0.925 &       & -2.046 & 0.909 & 1.181 \\
          & (0.040) & (0.050) & (0.025) &       & (0.205) & (0.482) & (0.936) \\
    \multicolumn{1}{c}{100} & -1.921 & 0.925 & 0.911 &       & -2.014 & 0.956 & 1.085 \\
          & (0.027) & (0.030) & (0.014) &       & (0.154) & (0.470) & (0.901) \\
    \multicolumn{1}{c}{250} & -1.945 & 0.946 & 0.919 &       & -2.016 & 0.944 & 1.109 \\
          & (0.030) & (0.030) & (0.014) &       & (0.140) & (0.464) & (0.901) \\
    \multicolumn{1}{c}{500} & -1.994 & 1.043 & 0.965 &       & -2.010 & 0.967 & 1.064 \\
          & (0.029) & (0.042) & (0.020) &       & (0.134) & (0.463) & (0.892) \\
    \bottomrule
    \end{tabular}
    \begin{minipage}{12cm}
    \vspace{0.1cm}
    \scriptsize
    \linespread{1}
    This table gives the mean of each estimator across the 1000 Monte Carlo estimates. The standard deviation of the Monte Carlo estimates is reported below the mean value of the point estimates in parentheses (this is a quantile-based estimate which uses the 0.05 and 0.95 quantiles of the Monte Carlo distribution of point estimates and the assumption of Normality).
    \end{minipage}
  \label{desig 3}%
\end{table}%

\section{Conclusion}
In this paper, I characterize the network formation process as a static game of incomplete
information, where the latent payoff of forming a link between two individuals depends on the structure of the network, as well as private information on agents’ attributes. I allow agents’
private unobserved attributes to be correlated with observables (i.e. existence of individual
fixed effects). Using data from a single large network, I propose a two-step estimator for the
model primitives. In the first step, I estimate agents’ equilibrium beliefs of other people’s
choice probabilities. In the second step, I plug in the first-step estimator to the conditional choice probability expression and estimate the model parameters and the unobserved individual fixed effects together using Joint MLE. Assuming that the observed attributes are discrete, I showed that the first step estimator is uniformly consistent with the rate $n^{-1/2}$, where $n$ is the number of individuals in the network. This rate corresponds to the usual $N^{-1/4}$ rate where $N$ stands for the total number of linking proposals and is the effective sample size. The slow convergence rate is translated to the second step so that the usual asymptotic bias of order $N^{-1/2}$ caused by the "incidental parameter problem" won't show up in the asymptotic distribution. The second-step estimator $\hat\beta$ subtracted by its mean converges asymptotically to a normal distribution at the rate $N^{-1/4}$. Monte Carlo Simulation shows that the estimator proposed in this paper performs well in finite samples.

\newpage
\setcounter{section}{1}
\section*{Appendix}
\subsection{Lemmas}
\indent The next two lemmas are to be used in the proofs of the asymptotics.
\begin{lemma}\label{G-Q}
    Under Assumptions 1,2 and 3,
\begin{align*}
    \sup_{1\leqslant i\leqslant n}\left|\frac{1}{n-1}\sum_{j\neq i}(G_{ij}-Q_{ij})\right|< O\left(\sqrt{\frac{\ln n}{n}}\right)
\end{align*}
with probability $1-O(n^{-2})$, and 
\begin{align*}
    \sup_{1\leqslant i\leqslant n}\left|\frac{1}{n-1}\sum_{j\neq i}(G_{ij}-\hat Q_{ij})\right|<O\left(\sqrt{\frac{\ln n}{n}}\right)
\end{align*}
with probability $1-O(n^{-2})$, where 
\begin{align*}
    Q_{ij}&:=Q_{ij}(\beta_0, A_{i0}, Z_{ij})\\
    \hat Q_{ij}&:=Q_{ij}(\beta_0, A_{i0}, \hat Z_{ij}).
\end{align*}
\end{lemma}
\begin{proof}
The first conclusion comes by applying Hoeffding's inequality
\begin{align*}
    Pr\left(\left|\frac{1}{n-1}\sum_{j \neq i}(G_{ij}-Q_{ij})\right|\geqslant\epsilon\right)\leqslant 2\exp\left(-\frac{2(n-1)\epsilon^2}{(1-2\kappa)^2}\right)
\end{align*}
for $\kappa$ as defined by (\ref{bdd}). Setting $\epsilon=\sqrt{\frac{3(1-2\kappa)^2}{2}\frac{\ln n}{n}}$ gives
\begin{align*}
    Pr&\left(\left|\frac{1}{n-1}\sum_{j \neq i}(G_{ij}-Q_{ij})\right|\geqslant\sqrt{\frac{3(1-2\kappa)^2}{2}\frac{\ln n}{n}}\right)\\
    &\leqslant 2\exp\left(-\frac{2(n-1)}{(1-2\kappa)^2}\frac{3(1-2\kappa)^2}{2}\frac{\ln n}{n}\right)\\
    &=2\exp\left(\ln\left(\frac{1}{n^3}\right)\frac{(n-1)}{n}\right)\\
    &=2\left(\frac{1}{n^3}\right)^{\frac{(n-1)}{n}}\\
    &=O(n^{-3}).
\end{align*}
Applying Boole's inequality then gives
\begin{align*}
    Pr&\left(\max_{1\leqslant i\leqslant n}\left|\frac{1}{n-1}\sum_{j \neq i}(G_{ij}-Q_{ij})\right|\geqslant \sqrt{\frac{3(1-2\kappa)^2}{2}\frac{\ln n}{n}}\right)\\
    &\leqslant n*O(n^{-3})\\
    &=O(n^{-2}),
\end{align*}
from which the first conclusion follows.\\
\indent To prove the second conclusion, first, observe that for any $i,j$
\begin{align*}
    \left|\frac{1}{n-1}\sum_{j\neq i}(G_{ij}-\hat Q_{ij})\right|\leqslant\left|\frac{1}{n-1}\sum_{j\neq i}(G_{ij}- Q_{ij})\right|+\left|\frac{1}{n-1}\sum_{j\neq i}(Q_{ij}-\hat Q_{ij})\right|.
\end{align*}
By the triangle inequality,
\begin{align*}
    \left|\frac{1}{n-1}\sum_{j\neq i}(Q_{ij}-\hat Q_{ij})\right|\leqslant\sup_{i,j}\left|Q_{ij}-\hat Q_{ij}\right|
\end{align*}
Applying mean value expansion gives that for any $i,j$
\begin{align*}
    \left|Q_{ij}-\hat Q_{ij}\right|&=\left|\frac{\exp(\bar Z_{ij}\beta_0+A_{i,0})\beta_0'}{(1+\exp(\bar Z_{ij}\beta_0+A_{i,0}))^2}(\hat Z_{ij}-Z_{ij})\right|\\
    &=O_p(1)O_p(n^{-1/2})\\
    &=O_p(n^{-1/2})
\end{align*}
where the second equality comes from condition (\ref{bdd}), Assumption \ref{Compact Support} and Lemma \ref{qhat}. The second conclusion follows from the first conclusion. 
\end{proof}
\begin{lemma} \label{Ahat-A}
    Under Assumptions \ref{independence},\ref{Equilibrium Selection} and \ref{Compact Support}, $\hat A_i(\beta_0)-A_i(\beta_0)$ has the asymptotically linear representation
\begin{align*}
    \hat A_i(\beta_0)-A_i(\beta_0)=\frac{\sum_{j \neq i}(G_{ij}-Q_{ij})}{\sum_{j \neq i}Q_{ij}(1-Q_{ij})}+\frac{\sum_{j \neq i}\hat Q_{ij}-Q_{ij}}{\sum_{j \neq i} Q_{ij}(1-Q_{ij})}+O_P\left(\frac{\ln n}{n}\right)
\end{align*}
\end{lemma}
\begin{proof}
Consider the first order condition with respect to $A$ 
\begin{align*}
    \frac{\partial\mathcal{L}_n(\beta_0, A, \hat q)}{\partial A}\bigg|_{A=\hat A(\beta_0)}=0;
\end{align*}
a mean value expansion gives that for all $i$
\begin{align}
     0=&\sum_{j\neq i}\left(G_{ij}- Q_{ij}(\beta_0, \hat A_i(\beta_0), \hat q_{ij})\right)\notag\\
     =&\sum_{j\neq i}\left(G_{ij}- Q_{ij}(\beta_0, A_i(\beta_0), \hat q_{ij})\right)\notag\\
     &-\sum_{j\neq i}(\hat A_i(\beta_0)-A_i(\beta_0))Q_{ij}(\beta_0, A_i(\beta_0), \hat q_{ij})[1-Q_{ij}(\beta_0, A_i(\beta_0), \hat q_{ij})]\notag\\
     &+\frac{1}{2}\sum_{j\neq i}(\hat A_i(\beta_0)-A_i(\beta_0))^2Q_{ij}(\beta_0, \bar A_i(\beta_0), \hat q_{ij})[1-Q_{ij}(\beta_0, \bar A_i(\beta_0), \hat q_{ij})][1-2Q_{ij}(\beta_0, \bar A_i(\beta_0), \hat q_{ij})].\label{expansionofa}
\end{align}
Denote the last term by $R_i$. The Triangle Inequality and Condition (\ref{bdd}) then implies
\begin{align}
    |R_i|\leqslant&\frac{1}{2}|\hat A_i(\beta_0)-A_i(\beta_0)|^2\sum_{j\neq i}\left|Q_{ij}(\beta_0, \bar A_i(\beta_0), \hat q_{ij})[1-Q_{ij}(\beta_0, \bar A_i(\beta_0), \hat q_{ij})][1-2Q_{ij}(\beta_0, \bar A_i(\beta_0), \hat q_{ij})]\right|\\
    \leqslant&\lambda_n^2O_p(n-1),\label{ri}
\end{align}
where $\lambda_n=\sup_{1\leqslant i\leqslant n}\left|\hat A_i-A_{i0}\right|\leqslant O_p(\sqrt{\frac{\ln n}{n}})$ according to Theorem \ref{uniform A}. From (\ref{expansionofa}) I have
\begin{align*}
    &\hat A_i(\beta_0)-A_i(\beta_0)\\
    =&\frac{\sum_{j \neq i}[G_{ij}-Q_{ij}(\beta_0, A_i(\beta_0), \hat q_{ij})]}{\sum_{j \neq i}Q_{ij}(\beta_0, A_i(\beta_0), \hat q_{ij})[1-Q_{ij}(\beta_0, A_i(\beta_0), \hat q_{ij})]}+\frac{R_i}{\sum_{j \neq i}Q_{ij}(\beta_0, A_i(\beta_0), \hat q_{ij})[1-Q_{ij}(\beta_0, A_i(\beta_0), \hat q_{ij})]}\\
    =&\frac{\sum_{j \neq i}[G_{ij}-Q_{ij}(\beta_0, A_i(\beta_0), \hat q_{ij})]}{\sum_{j \neq i}Q_{ij}(\beta_0, A_i(\beta_0), \hat q_{ij})[1-Q_{ij}(\beta_0, A_i(\beta_0), \hat q_{ij})]}+O_p\left(\frac{\sqrt{\ln n}}{n}\right)+O_p\left(\frac{\ln n}{n}\right)\\
    =&\frac{\sum_{j \neq i}(G_{ij}-Q_{ij})}{\sum_{j \neq i}Q_{ij}(1-Q_{ij})}+\frac{\sum_{j \neq i}\hat Q_{ij}-Q_{ij}}{\sum_{j \neq i} Q_{ij}(1-Q_{ij})}+O_p\left(\frac{\sqrt{\ln n}}{n}\right)
\end{align*}
where the second equality follows from (\ref{ri}) and Condition (\ref{bdd}) and the third equality come from a similar argument as in the proof of Lemma \ref{G-Q}. More specifically, from the proof of Lemma \ref{G-Q}, I know that $\hat Q_{ij}=Q_{ij}+O_p\left(\frac{1}{\sqrt n}\right)$, then applying Condition (\ref{bdd}) yields
\begin{align*}
    &\frac{\sum_{j \neq i}(G_{ij}-\hat Q_{ij})}{\sum_{j \neq i}\hat Q_{ij}(1-\hat Q_{ij})}\\
    =&\frac{\sum_{j \neq i}\left(G_{ij}- Q_{ij}\right)}{\sum_{j \neq i}\hat Q_{ij}(1-\hat Q_{ij})}+\frac{\sum_{j \neq i}\hat Q_{ij}-Q_{ij}}{\sum_{j \neq i} \hat Q_{ij}(1-\hat Q_{ij})}\\
    =&\frac{\sum_{j \neq i}\left(G_{ij}- Q_{ij}\right)}{\sum_{j \neq i} Q_{ij}(1- Q_{ij})}+\frac{\left(\left[\sum_{j \neq i}Q_{ij}(1-Q_{ij})\right]-\left[\sum_{j \neq i}\hat Q_{ij}(1-\hat Q_{ij})\right]\right)\sum_{j \neq i}\left(G_{ij}- Q_{ij}\right)}{\left[\sum_{j \neq i}Q_{ij}(1-Q_{ij})\right]\left[\sum_{j \neq i}\hat Q_{ij}(1-\hat Q_{ij})\right]}\\
    &+\frac{\sum_{j \neq i}\hat Q_{ij}-Q_{ij}}{\sum_{j \neq i} Q_{ij}(1-Q_{ij})}\\
    =&\frac{\sum_{j \neq i}\left(G_{ij}- Q_{ij}\right)}{\sum_{j \neq i} Q_{ij}(1- Q_{ij})}+\frac{\sum_{j \neq i}\hat Q_{ij}-Q_{ij}}{\sum_{j \neq i} Q_{ij}(1-Q_{ij})}+O_p\left(\frac{1}{n}\right),
\end{align*}
the conclusion thus follows.
\end{proof}

\subsection{Proofs}
\begin{proof}[\proofname\ of Lemma \ref{qhat}]
   As specified in (\ref{optimal strategy}), the optimal linking decision of agent $i$ with agent $j$ is $G_{ij}(X,A_{i}, \varepsilon_i, \sigma^*)$.
\begin{align*}
    &\sup_{s,t}\big|\hat q_{n, st}-q_{st}(X,\sigma^*)\big|\\
    =&\sup_{s,t}\Bigg|\frac{\sum_{i}\sum_{j\neq i}\big(G_{ij}-P(G_{ij}=1|X_{i}=x_s, X_{j}=x_t,X, \sigma^*)\big)\mathbf{1}\big\{X_{i}=x_s,X_{j}=x_t\big\}}{\sum_{i}\sum_{j\neq i}\mathbf{1}\big\{X_{i}=x_s,X_{j}=x_t\big\}}\Bigg|
\end{align*}
Denote the fraction term by $\Delta_{n, st}$. It suffices to show that  
\begin{align*}
\lim_{\eta\rightarrow\infty}\lim_{n\rightarrow\infty}P\bigg(\sup_{s,t}\big|\Delta_{n, st}\big|>\eta n^{-1/2}\bigg)=0.
\end{align*}
By the law of iterated expectations and dominated convergence theorem, it suffices to show
\begin{align*}
    P\bigg(\sup_{s,t}\big|\Delta_{n, st}\big|>\eta n^{-1/2}\bigg|X,\sigma^*\bigg)\xrightarrow{p}0\text{\ as\ } \eta, n \rightarrow \infty.
\end{align*}
Note that 
\begin{align*}
    P\bigg(\sup_{s,t}\big|\Delta_{n, st}\big|>\eta n^{-1/2}\bigg|X,\sigma^*\bigg)&\leqslant\sum_{s,t}P\big(\big|\Delta_{n, st}\big|>\eta n^{-1/2}\big|X,\sigma^*\big)\\
    &\leqslant\sum_{s,t}\frac{n\mathbb E\big(\Delta^2_{st}\big|X,\sigma^*\big)}{\eta^2}\\
    &\leqslant\frac{nT_x^2}{\eta^2}\max_{s,t}\mathbb E \big(\Delta^2_{st}\big|X,\sigma^*\big)
\end{align*}
Then it suffices to show $E \big(\Delta^2_{st}\big|X,\sigma^*\big)=O\big(n^{-1}\big)$ for all $s, t$.
\begin{align}
    &\mathbb E \left(\Delta^2_{st}\big|X,\sigma^*\right)\notag\\
    =&\frac{\sum_i\sum_{i\neq j}Var\left(G_{ij}|X_{i}=x_s, X_{j}=x_t, X, \sigma^*\right)\mathbf{1}\left\{X_{i}=x_s, X_{j}=x_t\right\}}{\left(\sum_i\sum_{j\neq i}\mathbf{1}\left\{X_{i}=x_s, X_{j}=x_t\right\}\right)^2}\notag\\
    +&\frac{\sum_i\sum_{i\neq j}\sum_{k\neq i,j}Cov\left(G_{ij}, G_{ik}|X_{i}=x_s, X_{j}=x_t, X, \sigma^*\right)\mathbf{1}\left\{X_{i}=x_s, X_{j}=x_t\right\}}{\left(\sum_i\sum_{j\neq i}\mathbf{1}\left\{X_{i}=x_s, X_{j}=x_t\right\}\right)^2}
\end{align}
where I used the fact that link proposals from different agents are independent, i.e. $G_{ij}(X, A_i, \varepsilon_i, \sigma)\perp G_{i'j'}(X, A_{i'}, \varepsilon_{i'}, \sigma^*)\big|X,\sigma$, so $Cov\big(G_{ij}, G_{i'j'}|X_i=X_{i'}=x_s, X_j=X_{j'}=x_t, X, \sigma^*\big)=0$ for all $i\neq i'$.\\
\indent Since $G_{ij}$ is a binary random variable, $Var\left(G_{ij}|X_i=x_s, X_j=x_t, X, \sigma^*\right)\leqslant\frac{1}{4}$. The first term is bounded by 
\begin{align*}
    \frac{1}{4}\left(\sum_i\sum_{j\neq i}\mathbf{1}\left\{X_i=x_s, X_j=x_t\right\}\right)^{-1}
\end{align*}
Then for the second term, by Cauchy-Schwarz Inequality,
\begin{align*}
    &Cov\left(G_{ij}, G_{ik}|X_i=x_s, X_j=X_k=x_t, X, \sigma^*\right)\\
    \leqslant & \ Var\left(G_{ij}\big|X_i=x_s, X_j=x_t, X, \sigma^*\right)^{1/2}Var\left(G_{ik}\big|X_i=x_s, X_k=x_t, X, \sigma^*\right)^{1/2}\\
    \leqslant &\ \frac{1}{4}
\end{align*}
so the second term is bounded by 
\begin{align*}
    &\frac{\frac{1}{4n}\frac{1}{n(n-1)(n-2)}\sum_i\sum_{j\neq i}\sum_{k\neq i,j}\mathbf{1}\left\{X_i=x_s, X_j=X_k=x_t\right\}}{\frac{n-1}{n-2}\left(\frac{1}{n(n-1)}\sum_i\sum_{j\neq i}\left\{X_i=x_s, X_j=x_t\right\}\right)^2}\\
    \leqslant&\frac{1}{4n}\frac{\frac{1}{n(n-1)(n-2)}\sum_i\sum_{j\neq i}\sum_{k\neq i,j}\mathbf{1}\left\{X_i=x_s, X_j=X_k=x_t\right\}}{\left(\frac{1}{n(n-1)}\sum_i\sum_{j\neq i}\left\{X_i=x_s, X_j=x_t\right\}\right)^2}
\end{align*}
Both the numerator and denominator are U-statistics. It's straightforward to show that they converge to their expectations. Therefore, the sum of the first and second terms are $O(\frac{1}{n})$, and the proof is complete. 
\end{proof}
\begin{proof}[\proofname\ of Theorem \ref{Consistency}]
    
\indent According to Assumption \ref{Joint FE Identification}, $\beta_0, A_0$ uniquely maximizes $\mathbb E[\mathcal L_n(\beta,A,q)|X,A_0]$. Since $(\hat\beta, \hat A)$ solves $\max_{\beta\in \mathbb B, A\in\mathbb A} \mathcal L_n(\beta, A, \hat q)$, it suffices to show that 
\begin{align}
    \sup_{\beta, A} \bigg|\mathcal L_n(\beta, A, \hat q)-\mathbb E[\mathcal L_n(\beta,A,q)|X,A_0]\bigg|\xrightarrow{p}0.\label{thm1}
\end{align}
By the triangle inequality, the left-hand side is less than or equal to
\begin{align*}
    \underbrace{\sup_{\beta, A} \bigg|\mathcal L_n(\beta, A, \hat q)-\mathbb E[\mathcal L_n(\beta,A,\hat q)|X,A_0]\bigg|}_{I}+\underbrace{\sup_{\beta, A} \bigg|\mathbb E[\mathcal L_n(\beta,A,\hat q)|X,A_0]-\mathbb E[\mathcal L_n(\beta,A,q)|X,A_0]\bigg|}_{II}
\end{align*}
By Continuous Mapping Theorem and Lemma \ref{qhat}, $II=o_p(1)$.
By the Logit formalization of $Q_{ij}(\beta, A_i, \hat q)$,
\begin{align*}
    I&=\sup_{\beta, A} \bigg|\mathcal L_n(\beta, A, \hat q)-\mathbb E[\mathcal L_n(\beta,A,\hat q)|X,A_0]\bigg|\\
    &=\sup_{\beta, A} \Bigg|\frac{1}{n(n-1)}\sum_i\sum_{j\neq i}(G_{ij}-Q_{ij})\ln\Big(\frac{Q_{ij}(\beta, A_i, \hat q)}{1-Q_{ij}(\beta, A_i, \hat q)}\Big)\Bigg|
\end{align*}
where $Q_{ij}:= Q_{ij}(\beta_0, A_0, \hat q)$\\
According to the Triangle Inequality,
\begin{align*}
    &\Bigg|\frac{1}{n(n-1)}\sum_i\sum_{j\neq i}(G_{ij}-Q_{ij})\ln\Big(\frac{Q_{ij}(\beta, A_i, \hat q)}{1-Q_{ij}(\beta, A_i, \hat q)}\Big)\Bigg|\\
    &\leqslant \frac{1}{n}\sum_i\Bigg|\frac{1}{n-1}\sum_{j\neq i}(G_{ij}-Q_{ij})\ln\Big(\frac{Q_{ij}(\beta, A_i, \hat q)}{1-Q_{ij}(\beta, A_i, \hat q)}\Big)\Bigg|
\end{align*}
Condition (\ref{bdd}) implies that $\ln(\frac{\kappa}{1-\kappa})\leqslant\ln\Big(\frac{Q_{ij}(\beta, A_i, \hat q)}{1-Q_{ij}(\beta, A_i, \hat q)}\Big)\leqslant\ln(\frac{1-\kappa}{\kappa})$, thus $(\kappa-1)\ln\frac{1-\kappa}{\kappa}\leqslant(G_{ij}-Q_{ij})\ln\Big(\frac{Q_{ij}(\beta, A_i, \hat q)}{1-Q_{ij}(\beta, A_i, \hat q)}\Big)\leqslant (1-\kappa)\ln\frac{1-\kappa}{\kappa}$. According to Hoeffding's inequality,
\begin{align*}
    Pr\Bigg(\Bigg|\frac{1}{n-1}\sum_{j\neq i}(G_{ij}-Q_{ij})\ln\Big(\frac{Q_{ij}(\beta, A_i, \hat q)}{1-Q_{ij}(\beta, A_i, \hat q)}\Big)\Bigg|\geqslant\epsilon\Bigg)\leqslant 2\exp\Bigg(-\frac{(n-1)\epsilon^2}{2(1-\kappa)^2(\ln\frac{1-\kappa}{\kappa})^2}\Bigg)
\end{align*}
Take $\epsilon=\sqrt{\frac{3\ln n}{n}}$, applying Boole's inequality, for any $\beta\in \mathbb B$, $A\in \mathbb A^n$,
\begin{align*}
     &Pr\Bigg(\max_{1\leqslant i\leqslant n}\Bigg|\frac{1}{n-1}\sum_{j\neq i}(G_{ij}-Q_{ij})\ln\Big(\frac{Q_{ij}(\beta, A_i, \hat q)}{1-Q_{ij}(\beta, A_i, \hat q)}\Big)\Bigg|\geqslant\sqrt{\frac{3\ln n}{n}}\Bigg)\\
     \leqslant&nPr\Bigg(\Bigg|\frac{1}{n-1}\sum_{j\neq i}(G_{ij}-Q_{ij})\ln\Big(\frac{Q_{ij}(\beta, A_i, \hat q)}{1-Q_{ij}(\beta, A_i, \hat q)}\Big)\Bigg|\geqslant\sqrt{\frac{3\ln n}{n}}\Bigg)\\ 
     \leqslant&\left(\frac{2}{n^2}\right)^{-\frac{(n-1)}{2n(1-\kappa)^2(\ln\frac{1-\kappa}{\kappa})^2}}\\
     =&O\big(\frac{1}{n^2}\big)\\
     \Longrightarrow&Pr\Bigg(\Bigg|\frac{1}{n(n-1)}\sum_{i}\sum_{j\neq i}(G_{ij}-Q_{ij})\ln\Big(\frac{Q_{ij}(\beta, A_i, \hat q)}{1-Q_{ij}(\beta, A_i, \hat q)}\Big)\Bigg|\geqslant\sqrt{\frac{3\ln n}{n}}\Bigg)\\
     \leqslant&Pr\Bigg(\frac{1}{n}\sum_{i}\Bigg|\frac{1}{n-1}\sum_{j\neq i}(G_{ij}-Q_{ij})\ln\Big(\frac{Q_{ij}(\beta, A_i, \hat q)}{1-Q_{ij}(\beta, A_i, \hat q)}\Big)\Bigg|\geqslant\sqrt{\frac{3\ln n}{n}}\Bigg)\\
     \leqslant&Pr\Bigg(\max_{1\leqslant i\leqslant n}\Bigg|\frac{1}{n-1}\sum_{j\neq i}(G_{ij}-Q_{ij})\ln\Big(\frac{Q_{ij}(\beta, A_i, \hat q)}{1-Q_{ij}(\beta, A_i, \hat q)}\Big)\Bigg|\geqslant\sqrt{\frac{3\ln n}{n}}\Bigg)\leqslant O\big(\frac{1}{n^2}\big),
\end{align*}
which implies the uniform convergence result:
\begin{align}
    Pr\Bigg(\sup_{\beta, A}\Bigg|\frac{1}{n(n-1)}\sum_i\sum_{j\neq i}(G_{ij}-Q_{ij})\ln\Big(\frac{Q_{ij}(\beta, A_i, \hat q)}{1-Q_{ij}(\beta, A_i, \hat q)}\Big)\Bigg|\geqslant\sqrt{\frac{3\ln n}{n}}\Bigg)\leqslant O\big(\frac{1}{n^2}\big)
\end{align}
and hence $I=o_p(1)$ and (\ref{thm1}) follows.
\end{proof}
\begin{proof}[\proofname\ of Theorem \ref{uniform A}]
    Let $A_0$ denote the population vector of heterogeneity terms and $A_1=\varphi(A_0)$. From (\ref{fixedpointsolution}), I have
\begin{align*}
    A_{1,i}-A_{0,i}=\ln\sum_{j\neq i}G_{ij}-\ln\sum_{j \neq i}\frac{\exp(\hat Z_{ij}'\hat\beta+A_{0i})}{1+\exp(\hat Z_{ij}'\hat\beta+A_{0i})}
\end{align*}
A Taylor expansion of the second term on the right-hand side gives:
\begin{align*}
    &\ln\sum_{j \neq i}\frac{\exp(\hat Z_{ij}'\hat\beta+A_{0i})}{1+\exp(\hat Z_{ij}'\hat\beta+A_{0i})}\\
    =&\ln\sum_{j \neq i}\frac{\exp(\hat Z_{ij}'\beta_0+A_{0i})}{1+\exp(\hat Z_{ij}'\beta_0+A_{0i})}+\frac{\sum_{j\neq i}Q_{ij}(\bar\beta, A_{i0}, \hat Z_{ij})(1-Q_{ij}(\bar\beta, A_{i0}, \hat Z_{ij}))\hat Z_{ij}'}{\sum_{i\neq i}Q_{ij}(\bar\beta, A_{i0}, \hat Z_{ij})}(\hat\beta-\beta_0)
\end{align*}
Using (\ref{bdd}), the compact support of $Z_{ij}$ , and Theorem \ref{Consistency},
\begin{align*}
    &\left|\frac{\sum_{j\neq i}Q_{ij}(\bar\beta, A_{i0}, \hat Z_{ij})(1-Q_{ij}(\bar\beta, A_{i0}, \hat Z_{ij}))\hat Z_{ij}'}{\sum_{i\neq i}Q_{ij}(\bar\beta, A_{i0}, \hat Z_{ij})}(\hat\beta-\beta_0)\right|\\
    \leqslant&\sum_{j\neq i}\left|\frac{Q_{ij}(\bar\beta, A_{i0}, \hat Z_{ij})(1-Q_{ij}(\bar\beta, A_{i0}, \hat Z_{ij}))\hat Z_{ij}'}{\sum_{i\neq i}Q_{ij}(\bar\beta, A_{i0}, \hat Z_{ij})}\right|\left|(\hat\beta-\beta_0)\right|\\
    \leqslant&\frac{\sup_{z\in\mathbb Z}|z'|}{4\kappa}\left|(\hat\beta-\beta_0)\right|\\
    =&O_p(1)\cdot o_p(1)\\
    =&o_p(1).
\end{align*}
I can conclude that
\begin{align*}
    A_{1,i}-A_{0,i}=\ln\sum_{j\neq i}G_{ij}-\ln\sum_{j \neq i}\frac{\exp(\hat Z_{ij}'\beta_0+A_{0i})}{1+\exp(\hat Z_{ij}'\beta_0+A_{0i})}+o_p(1).
\end{align*}
Denote $\hat Q_{ij}:=\frac{\exp(\hat Z_{ij}'\beta_0+A_{0i})}{1+\exp(\hat Z_{ij}'\beta_0+A_{0i})}$, a mean value expansion around $\hat Q_{ij}$ gives
\begin{align*}
    \ln\sum_{j\neq i}G_{ij}=\ln\sum_{j \neq i}\hat Q_{ij}+\frac{\sum_{j\neq i}G_{ij}-\hat Q_{ij}}{\lambda\sum_{j\neq i}G_{ij}+(1-\lambda)\sum_{j\neq i}\hat Q_{ij}},
\end{align*}
for some $\lambda\in(0,1)$. By (\ref{bdd}), for all $i$
\begin{align*}
    \left|\frac{\sum_{j\neq i}(G_{ij}-\hat Q_{ij})}{\lambda\sum_{j\neq i}G_{ij}+(1-\lambda)\sum_{j\neq i}\hat Q_{ij}}\right|&\leqslant\frac{\left|\sum_{j\neq i}(G_{ij}-\hat Q_{ij})\right|}{(n-1)(1-\lambda)\kappa}
\end{align*}
Lemma \ref{G-Q} then gives, with probability $1-O(n^{-2})$, the uniform bound
\begin{align*}
    \sup_{1\leqslant i \leqslant n}\left|A_{1,i}-A_{0,i}\right|<O\left(\sqrt{\frac{\ln n}{n}}\right)
\end{align*}
Then the conclusion follows by applying Lemma 4 in \cite{graham2017econometric}.
\end{proof}
\begin{proof}[\proofname\ of Theorem \ref{asymptotic normality}]
\textit{Step 1. Characterizing the probability limit of the Hessian of the concentrated log-likelihood.}\\
\indent First define the following notations. The Hessian matrix of the joint log-likelihood is given by
\begin{align*}
    H_n=\begin{pmatrix}H_{n,\beta\beta}&H_{n,\beta A}\\H'_{n,\beta A}&H_{n,AA}\end{pmatrix}
\end{align*}
where
\begin{align}
    H_{n,\beta\beta}&=-\sum_{i}\sum_{j\neq i}Z_{ij}Z{ij}'Q_{ij}(1-Q_{ij})\label{hbb}\\
    H'_{n,\beta A}&=-\begin{pmatrix}\sum_{j\neq 1}Q_{1j}(1-Q_{1j})Z_{1j}'\\\vdots\\\sum_{j\neq n}Q_{nj}(1-Q_{nj})Z_{nj}'\end{pmatrix}\label{hba}\\
    H_{n,AA}&=-\begin{pmatrix}\sum_{j\neq 1}Q_{1j}(1-Q_{1j})&\cdots&0\\\vdots&\ddots&\vdots\\0&\cdots&\sum_{j\neq n}Q_{nj}(1-Q_{nj})\end{pmatrix}\label{haa}
\end{align}
and $\hat H_{n,\beta\beta}$, $\hat H'_{n,\beta A}$, and $\hat H_{n,A A}$ are defined by (\ref{hbb}), (\ref{hba}), (\ref{haa}) respectively with $Z_{ij}$ replaced by $\hat Z_{ij}$.\\
\indent Following Amemiya (1985, pp. 125-127), the Hessian of the concentrated likelihood is 
\begin{align*}
    &\frac{\partial^2\mathcal{L}_n^c(\beta_0, \hat A(\beta_0), \hat q)}{\partial\beta\partial\beta'}=\sum_{i=1}^n\sum_{j\neq i}\frac{\partial}{\partial\beta'}s_{\beta,ij}(\beta_0, \hat A_i(\beta_0), \hat q_{ij})\\
    =&\hat H_{n,\beta\beta}-\hat H_{n,\beta A}\hat H^{-1}_{n,AA}\hat H'_{n,\beta A}\\
    =&-\sum_{i=1}^n\sum_{j\neq i}\hat Z_{ij}\hat Z_{ij}'\hat Q_{ij}(1-\hat Q_{ij})+\sum_{i=1}^n\frac{\left(\sum_{j\neq i}\hat Q_{ij}(1-\hat Q_{ij})\hat Z_{ij}\right)\left(\sum_{j\neq i}\hat Q_{ij}(1-\hat Q_{ij})\hat Z_{ij}'\right)}{\sum_{j\neq i}\hat Q_{ij}(1-\hat Q_{ij})},
\end{align*}
which implies
\begin{align}
    &\frac{1}{n(n-1)}\sum_{i=1}^n\sum_{j\neq i}\frac{\partial}{\partial\beta'}s_{\beta,ij}(\beta_0, \hat A_i(\beta_0), \hat q_{ij})\notag\\
    =&-\frac{1}{n(n-1)}\sum_{i=1}^n\sum_{j\neq i}\hat Z_{ij}\hat Z_{ij}'\hat Q_{ij}(1-\hat Q_{ij})\notag\\
    &+\frac{1}{n}\sum_{i=1}^n\frac{\left(\frac{1}{n-1}\sum\limits_{j\neq i}\hat Q_{ij}(1-\hat Q_{ij})\hat Z_{ij}\right)\left(\frac{1}{n-1}\sum\limits_{j\neq i}\hat Q_{ij}(1-\hat Q_{ij})\hat Z_{ij}'\right)}{\frac{1}{n-1}\sum_{j\neq i}\hat Q_{ij}(1-\hat Q_{ij})}\notag\\
   =&-\frac{1}{n(n-1)}\sum_{i=1}^n\sum_{j\neq i}Z_{ij}Z_{ij}'Q_{ij}(1-Q_{ij})\notag\\
   &+\frac{1}{n}\sum_{i=1}^n\frac{\left(\frac{1}{n-1}\sum\limits_{j\neq i}Q_{ij}(1-Q_{ij})Z_{ij}\right)\left(\frac{1}{n-1}\sum\limits_{j\neq i}Q_{ij}(1-Q_{ij})Z_{ij}'\right)}{\frac{1}{n-1}\sum_{j\neq i}Q_{ij}(1-Q_{ij})}+o_p(1)\notag\\
   =&\mathcal{I}_0+o_p(1), \label{plimhessian}
\end{align}
where $\mathcal{I}_0$ is as defined in (\ref{i0}). The second equality in (\ref{plimhessian}) is given by the same logic as the proof of Lemma \ref{G-Q} and more involved calculations.\\
\\
\textit{Step 2. Asymptotic Linear Representation}\\
\indent Consider the first-order condition associated with the concentrated log-likelihood
\begin{align*}
    \frac{\partial\mathcal{L}_n^c(\beta, \hat A(\beta), \hat q)}{\partial\beta}\bigg|_{\beta=\hat\beta}=0;
\end{align*}
a mean value expansion gives
\begin{align*}
    0=\sum_{i=1}^n\sum_{j\neq i}s_{\beta, ij}(\hat\beta, \hat A_i(\hat \beta), \hat q_{ij})=\sum_{i=1}^n\sum_{j\neq i}s_{\beta, ij}(\beta_0, \hat A_i(\beta_0), \hat q_{ij})+\sum_{i=1}^n\sum_{j\neq i}\frac{\partial}{\partial\beta'}s_{\beta,ij}(\bar\beta, \hat A_i(\bar\beta), \hat q_{ij})(\hat \beta-\beta_0),
\end{align*}
which implies
\begin{align*}
    &\sqrt{n}(\hat \beta-\beta_0)\\
    =&\underbrace{-\left[\frac{1}{n^2}\sum_{i=1}^n\sum_{j\neq i}\frac{\partial}{\partial\beta'}s_{\beta,ij}(\bar\beta, \hat A_i(\bar\beta), \hat q_{ij})\right]^{-1}}_{I^{-1}}\underbrace{\left[\frac{1}{n^{3/2}}\sum_{i=1}^n\sum_{j\neq i}s_{\beta, ij}(\beta_0, \hat A_i(\beta_0), \hat q_{ij})\right]}_{II}.
\end{align*}
\indent The first term $I$ converges in probability to $\mathcal I_0$ as defined in (\ref{i0}). I cannot apply a CLT directly to $II$ because of the strong correlation between summands caused by using the same set of data to get $\hat q$, $\hat A$ and estimator of $\beta$. \\
\indent A second order Taylor expansion of $II$ gives
\begin{small}
\begin{align}
    &\frac{1}{n^{3/2}}\sum_{i=1}^n\sum_{j\neq i}s_{\beta, ij}(\beta_0, \hat A_i(\beta_0), \hat q_{ij})\notag\\
    =&\frac{1}{n^{3/2}}\sum_{i=1}^n\sum_{j\neq i}s_{\beta, ij}(\beta_0, A_i(\beta_0), q_{ij})\notag\\
    &-\frac{1}{n^{3/2}}\sum_{i=1}^n\sum_{j\neq i}(\hat A_i(\beta_0)-A_i(\beta_0))Q_{ij}(1-Q_{ij})Z_{ij}\notag\\
    &-\frac{1}{n^{3/2}}\sum_{i=1}^n\sum_{j\neq i}Q_{ij}(1-Q_{ij})Z_{ij}\beta_0'(\hat Z_{ij}-Z_{ij})+Q_{ij}(\hat Z_{ij}-Z_{ij})\notag\\
    &-\frac{1}{2}\frac{1}{n^{3/2}}\sum_{i=1}^n\sum_{j\neq i}(\hat A_i(\beta_0)-A_i(\beta_0))^2\bar Q_{ij}(1-\bar Q_{ij})(1-2\bar Q_{ij})\bar Z_{ij}\notag\\
    &-\frac{1}{2}\frac{1}{n^{3/2}}\sum_{i=1}^n\sum_{j\neq i}(\hat Z_{ij}-Z_{ij})'\nabla_{Z_{ij}Z_{ij}'}s_{\beta, ij}(\beta_0,\bar A_i(\beta_0), \bar q_{ij})(\hat Z_{ij}-Z_{ij})\notag\\
    &-\frac{1}{n^{3/2}}\sum_{i=1}^n\sum_{j\neq i}(\hat A_i(\beta_0)-A_i(\beta_0))\left[\bar Q_{ij}(1-\bar Q_{ij})(1-2\bar Q_{ij})\bar Z_{ij}\beta_0'(\hat Z_{ij}-Z_{ij})+\bar Q_{ij}(1-\bar Q_{ij})(\hat Z_{ij}-Z_{ij})\right]
    \label{2ndorderexpantion},
\end{align}
\end{small}
where $\bar Q_{ij}=\frac{\exp(\bar Z_{ij}\beta_0+ \bar A_i)}{1+\exp(\bar Z_{ij}\beta_0+\bar A_i)}$, with $\bar A_i$ between $\hat A$ and $A$, $\bar Z_{ij}$ between $\hat Z_{ij}$ and $Z_{ij}$, for all $i,j$.\\
The main result follows by showing that 
\begin{enumerate}[(i)]
    \item A CLT can be applied to the second and third terms of (\ref{2ndorderexpantion}).
    \item The first term converges in probability to 0.
    \item The last three terms (second-order terms) converge in probability to 0. 
\end{enumerate}
\indent I start from the last three terms in $(\ref{2ndorderexpantion})$. Condition (\ref{bdd}), compact support and Theorem \ref{uniform A} implies that
\begin{align*}
    &\left|-\frac{1}{2}\frac{1}{n^{3/2}}\sum_{i=1}^n\sum_{j\neq i}(\hat A_i(\beta_0)-A_i(\beta_0))^2\bar Q_{ij}(1-\bar Q_{ij})(1-2\bar Q_{ij})\bar Z_{ij}\right|\\
    \leqslant&\frac{1}{2}\sqrt{n}\left|\lambda_n\right|^2O_p\left(\frac{\ln n}{n}\right)\\
    =&O_p\left(\frac{\ln n}{\sqrt{n}}\right)\\
    =&o_p(1).
\end{align*}
By the same argument, it can be shown that
\begin{small}
\begin{align*}
    &\left|-\frac{1}{2}\frac{1}{n^{3/2}}\sum_{i=1}^n\sum_{j\neq i}(\hat Z_{ij}-Z_{ij})'\nabla_{Z_{ij}Z_{ij}'}s_{\beta, ij}(\beta_0,\bar A_i(\beta_0), \bar q_{ij})(\hat Z_{ij}-Z_{ij})\right|=o_p(1)\\
    &\left|-\frac{1}{n^{3/2}}\sum_{i=1}^n\sum_{j\neq i}(\hat A_i(\beta_0)-A_i(\beta_0))\left[\bar Q_{ij}(1-\bar Q_{ij})(1-2\bar Q_{ij})\bar Z_{ij}\beta_0'(\hat Z_{ij}-Z_{ij})+\bar Q_{ij}(1-\bar Q_{ij})(\hat Z_{ij}-Z_{ij})\right]\right|\\
    &=o_p(1).
\end{align*}
\end{small}
\indent Then I consider the first term in (\ref{2ndorderexpantion}). By Lemma \ref{G-Q}, 
\begin{align*}
    &\frac{1}{n^{3/2}}\sum_{i=1}^n\sum_{j\neq i}s_{\beta, ij}(\beta_0, A_i(\beta_0), q_{ij})\\
    =&\frac{1}{n^{3/2}}\sum_{i=1}^n\sum_{j\neq i}(G_{ij}-Q_{ij})Z_{ij}\\
    \leqslant&\frac{1}{\sqrt{n}}\sup_{Z\in\mathbb{Z}}|Z|\frac{1}{n}\sum_{i=1}^n\sum_{j\neq i}(G_{ij}-Q_{ij})\\
    \leqslant&\frac{1}{\sqrt{n}}O_p\left(1\right)\\
    =&o_p(1).
\end{align*}
where the second inequality comes from the fact that $\frac{1}{\sqrt{n(n-1)}}\sum_{i=1}^n\sum_{j\neq i}(G_{ij}-Q_{ij})=O_p(1)$. This is true because $G_{ij}$ are independent conditional on $A$ and $X$. Applying the central limit theorem yields the desired conclusion.\\
\indent Then look at the second term. Applying Lemma \ref{G-Q} and \ref{Ahat-A} yields
\begin{align}
    &\frac{1}{n^{3/2}}\sum_{i=1}^n\sum_{j\neq i}(\hat A_i(\beta_0)-A_i(\beta_0))Q_{ij}(1-Q_{ij})Z_{ij}\notag\\
    =&\frac{1}{n^{3/2}}\sum_{i=1}^n\sum_{j\neq i}\left( \frac{\sum_{j \neq i}(G_{ij}-Q_{ij})}{\sum_{j \neq i}Q_{ij}(1-Q_{ij})}+\frac{\sum_{j \neq i}\hat Q_{ij}-Q_{ij}}{\sum_{j \neq i} Q_{ij}(1-Q_{ij})}+O_P\left(\frac{1}{n}\right)\right)Q_{ij}(1-Q_{ij})Z_{ij}\notag\\
    =&\frac{1}{n^{3/2}}\sum_{i=1}^n\sum_{j\neq i}\left(\frac{\sum_{j \neq i}\hat Q_{ij}-Q_{ij}}{\sum_{j \neq i} Q_{ij}(1-Q_{ij})}\right)Q_{ij}(1-Q_{ij})Z_{ij}+o_p(1)\notag\\
    =&\frac{1}{n^{3/2}}\sum_{i=1}^n\left(\frac{\sum_{j\neq i}Q_{ij}(1-Q_{ij})Z_{ij}}{\sum_{j \neq i} Q_{ij}(1-Q_{ij})}\right)\sum_{j \neq i}\left(\frac{\exp(Z_{ij}'\beta_0+A_i(\beta_0))}{1+\exp(Z_{ij}'\beta_0+A_i(\beta_0))}(\hat Z_{ij}-Z_{ij})\right)+o_p(1)\notag.
\end{align}
The sum of the second and third terms can be written as
\begin{align}
    =&\frac{1}{n^{3/2}}\sum_{i=1}^n\sum_{j \neq i}M_{ij}(\hat Z_{ij}-Z_{ij})+o_p(1),\label{2nd3rdterm}
\end{align}
where $M_{ij}=\left(\frac{\sum_{j\neq i}Q_{ij}(1-Q_{ij})Z_{ij}}{\sum_{j \neq i} Q_{ij}(1-Q_{ij})}\right)\frac{\exp(Z_{ij}'\beta_0+A_i(\beta_0))}{1+\exp(Z_{ij}'\beta_0+A_i(\beta_0))}I_d+\left[Q_{ij}(1-Q_{ij})Z_{ij}\beta_0'+Q_{ij}I_d\right]$. \\
\indent Define $\zeta_{ij}=(W_{ij}', G_{ji},\frac{1}{n-2}\sum_{k\neq i,j} G_{jk})'$. As defined in Section 3, $\hat Z_{ij}:=(W_{ij}', \hat q_{ji},  \frac{1}{n-2}\sum_{k\neq i,j} \hat q_{jk})'$. I will show that
\begin{align}
    \frac{1}{n^{3/2}}\sum_{i=1}^n\sum_{j \neq i}M_{ij}(\hat q_{ij}-G_{ij})=0,\label{replace}
\end{align}
so that 
\begin{align*}
    \frac{1}{n^{3/2}}\sum_{i=1}^n\sum_{j \neq i}M_{ij}(\hat Z_{ij}-\zeta_{ij})=0,
\end{align*}
and hence I can replace $\hat Z_{ij}$ in (\ref{2nd3rdterm}) with $\zeta_{ij}$. To see why (\ref{replace}) holds, observe that
\begin{align*}
    \frac{1}{n^{3/2}}\sum_{i=1}^n\sum_{j \neq i}M_{ij}\left(\frac{\sum_{k=1}^n\sum_{l\neq k}G_{ij}\mathbf{1}\{W_{k,l}=W_{ij}\}}{\sum_{k=1}^n\sum_{l\neq k}\mathbf{1}\{W_{k,l}=W_{ij}\}}-G_{ij}\right)=0,
\end{align*}
and the claim follows. Define $V_{i}=\frac{1}{n}\sum_{j \neq i}M_{ij}(\zeta_{ij}-Z_{ij})$. Then (\ref{2nd3rdterm}) can be written as
\begin{align*}
    \frac{1}{\sqrt{n}}\sum_{i=1}^nV_i+o_p(1).
\end{align*}
where $\left\{V_i\right\}_{i=1}^n$ is independently distributed, conditional on $X,\sigma$. \\
\\
\textit{Step 3. Demonstration of Asymptotic Normality of the second and third term in (\ref{2ndorderexpantion}).}\\
\indent To apply CLT, I need to check the Lindeberg condition. Take any vector $a\in \mathbb{R}^d$, the conditional mean of $\frac{1}{\sqrt{n}}a'V_i$ is 0
\begin{align*}
    \mathbb{E}\left[\frac{1}{\sqrt{n}}a'V_i\bigg|X,\sigma\right]=0,
\end{align*}
The conditional variance of $\frac{1}{\sqrt{n}}\sum_i a'V_i$ given $X,\sigma$ is
\begin{align*}
    Var\left(\frac{1}{\sqrt{n}}\sum_i a'V_i\Bigg|X, \sigma\right)=\frac{1}{n}\sum_i\mathbb{E}\left[(a'V_i)^2|X, \sigma\right]:=\Omega_n.
\end{align*}
By compact support, Condition (\ref{bdd}), and Lemma \ref{qhat}, 
\begin{align*}
    \frac{\max_{i}|\frac{1}{\sqrt{n}}a'V_i|}{\sqrt{\Omega_n}}\xrightarrow{p}0.
\end{align*}
To check the Lindeberg condition, note that for any $\epsilon>0$
\begin{align*}
    &\frac{1}{\Omega_n}\sum_{i}\mathbb{E}\left[\frac{1}{n}(a'V_i)^2\mathbf{1}\left\{\frac{|\frac{1}{\sqrt{n}}a'V_i|}{\sqrt{\Omega_n}}>\epsilon\right\}\bigg| X, \sigma\right]\\
    \leqslant&\frac{1}{\Omega_n}\sum_{i}\mathbb{E}\left[\frac{1}{n}(a'V_i)^2\mathbf{1}\left\{\frac{\max_i|\frac{1}{\sqrt{n}}a'V_i|}{\sqrt{\Omega_n}}>\epsilon\right\}\bigg| X, \sigma\right]\\
    \leqslant&\frac{1}{\Omega_n}\sum_{i}\mathbb{E}\left[\frac{1}{n}(a'V_i)^2\bigg| X, \sigma\right]=1.
\end{align*}
By the dominated convergence theorem, the Lindeberg condition follows, i.e. for any $\epsilon>0$
\begin{align*}
    \frac{1}{\Omega_n}\sum_{i}\mathbb{E}\left[\frac{1}{n}(a'V_i)^2\mathbf{1}\left\{\frac{|\frac{1}{\sqrt{n}}a'V_i|}{\sqrt{\Omega_n}}>\epsilon\right\}\bigg| X, \sigma\right]\xrightarrow{p}0.
\end{align*}
By Lindeberg-Feller CLT, for any $a\in \mathbb{R}^d$, 
\begin{align*}
    \Omega_n^{-1/2}\frac{1}{\sqrt{n}}\sum_ia'V_i\xrightarrow{d}N(0,1)
\end{align*}
\indent Combining with the result in \textit{Step 1}, this yields the desired conclusion that for any $a\in\mathbb{R}^d$,
\begin{align*}
    \frac{\sqrt{n}a'(\hat \beta-\beta_0)}{\|a\|^{-1/2}(a'\mathcal{I}_0^{-1}\Omega_n\mathcal{I}_0^{-1}a)^{1/2}}\xrightarrow{d}N(0,1).
\end{align*}    
\end{proof}

\bibliographystyle{econ}
\bibliography{summer}

\end{document}